\documentclass{article}
\usepackage{epsfig}
\usepackage{graphicx}
\usepackage{mathrsfs}
\usepackage{graphicx}
\usepackage{amsmath}
\usepackage{amsfonts}
\usepackage{amssymb}
\usepackage{amsthm}
\usepackage{amscd}

\newtheorem{theorem}{Theorem}[section]
\newtheorem{lemma}{Lemma}[section]

\newtheorem{definition}{Definition}[section]

\begin{document}
\title{A Mathematical Aspect of Hohenberg-Kohn Theorem\thanks{{This work was supported by the National Science
       Foundation of China under grants 9133202 and 11671389, and the Key Research Program of Frontier Sciences of the Chinese Academy of Sciences under grant QYZDJ-SSW-SYS010.}}}

\author{ Aihui
Zhou\thanks{LSEC, Institute of Computational Mathematics and Scientific/Engineering Computing,
Academy of Mathematics and Systems Science, Chinese Academy of Sciences,  Beijing 100190, China; and School of Mathematical Sciences,
University of Chinese Academy of Sciences, Beijing 100049, China
({azhou@lsec.cc.ac.cn}).}}

\date{}
\maketitle

\begin{abstract}
The Hohenberg-Kohn theorem plays a fundamental role in density functional theory,
which  has become a basic tool for the
study of electronic structure of matter.
In this article, we study the Hohenberg-Kohn theorem
for a class of external potentials based on a unique continuation principle.
\end{abstract}

\vskip 0.2cm


\section{Introduction}\setcounter{equation}{0}

Density functional theory (DFT) is the most widely used many-body approach for electronic structure calculations and has significantly impacted on modern science and engineering (see, e.g., \cite{hohenberg-kohn64,martin04,noorden-maher-nuzz014,parr-yang89,redner05,naturematerials}),  of which Hohenberg-Kohn
theorem lies at the heart.

We see that the development of materials science, quantum chemistry, molecular biology,
and condensed-matter physics require to describe and understand the many-particle systems of thousands and hundreds of electrons
and  nuclei. We note that the basic mathematical model for electronic structure is  the Schr{\" o{dinger equation, which is set for the
many-particle wavefunction of the system in high dimensions. However, it is intractable to solve the Schr{\" o}dinger equation directly even using the modern advanced supercomputers except for few systems. Instead of using wavefunctions in high dimensions, DFT applies the particle density in three dimensions to model the system.  Derived from DFT, the so-called  Kohn-Sham model \cite{kohn-sham65}, which is equivalent to the Schr{\" o}dinger equation, is tractable. We understand that the theoretical basis of DFT is the Hohenberg-Kohn theorem.

Due to its subtleties, however, DFT is  not entirely elaborated yet
(c.f., e.g., \cite{kryachko05,kryachko06,kryachko-ludena14,levy10,lieb83,pino-bokanowski-etal07,szczepanik-dulak-wesolowski07,
zhou12,zhou15}
and references cited therein). For instance, the existing proofs of Hohenberg-Kohn theorem usually assume directly
or indirectly that
the particle density functional is differentiable or the particle density is positive  or the particle wavefunction does not vanish on a set of positive measure. We see that the particle density functional is not G$\hat{a}$teaux differentiable \cite{kvaal-helgaker15,lammert07,lieb83}, whether  the  particle density is positive is still open \cite{fournais-hoffmann-0stenhof08,zhou15}, and that the particle
wavefunction does not vanish on a set of positive measure  is
unclear in a real system (c.f. \cite{pino-bokanowski-etal07,zhou15}).
 We refer to \cite{ayers-golden-levy06,hadjisavvas-theophilou84,hohenberg-kohn64,kryachko05,kryachko-ludena14,levy79,levy10,lieb83,
pino-bokanowski-etal07,szczepanik-dulak-wesolowski07,zhou12} and references
cited therein for discussions on the Hohenberg-Kohn theorem.

In this paper, we prove the Hohenberg-Kohn theorem for  external
potentials that the associated wavefunctions do not vanish in an open set. We understand from Theorem XIII.57 of \cite{reed-simon78} together with its comment just below that the wavefunction of electronic Schr{\" o}dinger
equation does not vanishes on an open set, from which we obtain that Hohenberg-Kohn theorem holds true for Coulomb
potentials \cite{zhou12}, too. We note that the wavefunction  does not vanish in an open set is relatively mild
assumption \cite{reed-simon78,wolff93}.

\section{Ground state and density}\label{ground-state}\setcounter{equation}{0}
The
Hamiltonian of a system of interacting particles in an external potential $v$,
including any problem of electrons and  nuclei, is written as
\begin{eqnarray*}\label{hamit-ext}
{\cal H}&=&-\sum_{i=1}^N\frac{\hbar^2}{2m_e}\nabla^2_{x_i}
+\sum^N_{i=1}v(x_i)+\frac{1}{2}\sum_{i,j=1, i\ne
j}^N\frac{e^2}{|x_i-x_j|},
\end{eqnarray*}
where $\hbar$ is Planck's constant divided by $2\pi, m_e$ is the
mass of the electron, $\{x_i:i=1,\cdots,N\}$ are the variables that
describe the electron positions, and $e$ is the electronic charge.
Let $$T=-\sum_{i=1}^N\frac{\hbar^2}{2m_e}\nabla^2_{x_i}$$ be the
kinetic energy operator,$$ V_{ee}=\frac{1}{2}\sum_{i,j=1, i\ne
j}^N\frac{e^2}{|x_i-x_j|}$$ be the electron-electron repulsion
energy operator, and
\begin{eqnarray*}\label{density-def}
\rho(x)\equiv
\rho^\Psi(x)=N\sum_{\sigma_1,\sigma_2,\cdots,\sigma_N}\int_{\mathbb{R}^{3(N-1)}}
|\Psi((x,\sigma_1),(x_2,\sigma_2),\cdots,(x_N,\sigma_N))|^2\textup{d}x_2\cdots
\textup{d}x_N \end{eqnarray*} be the single-particle density. We have the energy of the system as follows
\begin{eqnarray*}\label{energy}
E=(\Psi, {\cal H}\Psi)=(\Psi, (T+V_{ee})\Psi)+\int_{\mathbb{R}^3}
v(x)\rho(x)\textup{d}x.\end{eqnarray*}
We understand that for an electronic Coulomb system, the external potential
\begin{eqnarray*}\label{coulomb-v}
v(x)=-\sum_{j=1}^M\frac{Z_je^2}{|x-r_j|}
\end{eqnarray*} is determined by $\{Z_j:j=1,2,\cdots,M\}$, which are the valence
charges of the nuclei, and $\{r_j: j=1,2,\cdots,M\}$, which are the
positions of the nuclei.

Let ${\cal H}_0=T+V_{ee}$ and  $v\in L^{3/2}(\mathbb{R}^{3})+L^{\infty}(\mathbb{R}^3)$. The total
Hamiltonian is ${\cal H}_v={\cal H}_0+{\cal V}$, where
$${\cal V}=\sum_{i=1}^Nv(x_i).$$
The associated ground state energy $E(v)$ is
\begin{eqnarray}\label{ground-state}
E(v)\equiv E(v,N)=\inf\{(\Psi,{\cal H}_v\Psi): \Psi\in
\mathscr{W}_N\},
\end{eqnarray} where
$$\mathscr{W}_N=\{\Psi\in H^1(\mathbb{R}^{3N}):
\sum_{\sigma_1,\sigma_2,\cdots,\sigma_N}\int_{\mathbb{R}^{3N}}|\Psi|^2\textup{d}x_1\cdots\textup{d}x_N=1\}.$$

Since there may or may not be a minimizer $\psi$ in
$\mathscr{W}_N$ and if there is one it may not be unique
\cite{lieb83}, we introduce a set  of
minimizers\begin{eqnarray*} \mathscr {G}_{v}&\equiv& \mathscr
{G}_{v,N}=\arg\inf\{(\Psi,{\cal H}_v\Psi): \Psi\in \mathscr{W}_N\},\\[0.2cm]
 \mathcal{V}_N&=&\{v\in L^{3/2}(\mathbb{R}^{3})+L^{\infty}(\mathbb{R}^3):\mathscr
    {G}_{v}\neq\varnothing\}.
\end{eqnarray*}
Any $\Psi$ in $\mathscr {G}_{v}$ is called a ground state of
(\ref{ground-state}). If $\Psi\in \mathscr {G}_{v}$, then
\begin{eqnarray}\label{ground-state-weak}
\mathcal {H}_v\Psi=E(v)\Psi
\end{eqnarray}
in the distributional sense. We see from the Schr{\" o}dinger equation (\ref{ground-state-weak}) that
the density $\rho$ of ground state  is determined by
the external potential $v$.

\section{Proof of Hohenberg-Kohn theorem}\label{h-k-theorem}\setcounter{equation}{0}
The following  property, whose proof can be found in \cite{zhou12} (c.f. also \cite{eschrig03,levy82,lieb83,pino-bokanowski-etal07}), is crucial and useful:

\begin{lemma}\label{pino-etal-lemma}
Given $v,v'\in L^{3/2}(\mathbb{R}^{3})+L^{\infty}(\mathbb{R}^{3})$.  Let $\rho_v=\rho^{\Psi_v}$ and
$\rho_{v'}=\rho^{\Psi_{v'}}$ with $\Psi_v\in \mathscr {G}_{v}$ and
$\Psi_{v'}\in \mathscr {G}_{v'}$. If $\rho_v=\rho_{v'}$, then
\begin{eqnarray*}\label{distribution-identity}
\left(\sum^N_{i=1}(v'-v)(x_i)-(E(v')-E(v))\right)\Psi_v=0 ~ ~\mbox{a.e.}~ (x_1,x_2,\cdots,x_N)\in \mathbb{R}^{3N}.
\end{eqnarray*}
\end{lemma}

In our analysis, we need the following helpful conclusion, which is a modified version of Lemma 1 of \cite{pino-bokanowski-etal07}
\begin{lemma}\label{multivarible-identity}
Let $(a,b)\subset \mathbb{R}$. If $w\in L^1((a,b)^{3})$ satisfies
\begin{eqnarray}\label{multivarible-identity1}
\sum_{j=1}^Nw(x_j)=C ~ ~\mbox{a.e.}~ (x_1,x_2,\cdots,x_N)\in (a,b)^{3N}
\end{eqnarray}
for some constant $C$, then $w(x)=C/N ~ ~\mbox{a.e.}~ x\in (a,b)^3$.
\end{lemma}

\begin{proof}
 We use the similar argument as that in \cite{pino-bokanowski-etal07}. Let $J_\epsilon(x)=\epsilon^{-3}J(x/\epsilon)$, where
\begin{eqnarray*}
J(x)=\left\{\begin{array}{rcl}
 \mathcal{J}\exp\left(-1/(1-|x|^2)\right)& & \mbox{if}~~|x|<1,\\[0.2cm]
0 & & \mbox{if}~~|x|\ge 1,
\end{array}\right.
\end{eqnarray*}
with
$$
\mathcal{J}=1/\int_{\mathbb{R}^3}\exp\left(-1/(1-|x|^2)\right).
$$
We see that $J_\epsilon(x)$ is nonnegative, belongs to $C_0^{\infty}(\mathbb{R}^3)$, satisfies
$$
\int_{\mathbb{R}^3}J_\epsilon(x)\textup{d}x=1,
$$ and $J_\epsilon *{\tilde w}\in C^{\infty}(\mathbb{R}^3)$ (see e.g. page 36 of \cite{adams-fournier03}), where
$$
(J_\epsilon *{\tilde w})(x)=\int_{\mathbb{R}^3}J_\epsilon(x-y){\tilde w}(y)\textup{d}y
$$
and
\begin{eqnarray*}
{\tilde w}(x)=\left\{\begin{array}{rcl}
 w& & \mbox{if}~~x\in (a,b)^3,\\[0.2cm]
0 & & \mbox{if}~~x\in \mathbb{R}^3\setminus (a,b)^3.
\end{array}\right.
\end{eqnarray*}
Therefore, we obtain from (\ref{multivarible-identity1}) that
\begin{eqnarray}\label{multivarible-identity2}
\sum_{j=1}^N(J_\epsilon*{\tilde w})(x_j)=C ~ ~\forall~ (x_1,x_2,\cdots,x_N)\in (a,b)^{3N}.
\end{eqnarray}
Setting $x_j=x(j=1,2,\cdots,N)$ in (\ref{multivarible-identity2}), we get
\begin{eqnarray*}\label{multivarible-identity2}
(J_\epsilon *{\tilde w})(x)=C/N ~ ~\forall~ x\in (a,b)^3.
\end{eqnarray*}
Using the fact that $(J_\epsilon *{\tilde w})(x)$ converges to ${\tilde w}(x) ~a.e.$ (see e.g. Theorem 2.29 of \cite{adams-fournier03}) and ${\tilde w}=w$ in $(a,b)^3$, we then arrive at the conclusion and complete the proof.
\end{proof}

To restate and prove the Hohenberg-Kohn theorem, we recall the  unique continuation principle, which is significant in the context of partial differential equations
(see, e.g., \cite{jerison-kenig85,reed-simon78,wolff93} and references cited therein).

\begin{definition}
Equation (\ref{ground-state-weak}) has the unique continuation property if every solution that vanishes on an open
set of $\mathbb{R}^{3N}$ vanishes identically.
\end{definition}


Let $\mathbb{V}$ be the set of all the external potentials $v\in L^{3/2}(\mathbb{R}^{3})+L^{\infty}(\mathbb{R}^3)$
such that the  unique continuation principle holds true for  (\ref{ground-state-weak}).\footnote{We conjecture that (c.f.  Remarks (ii) of \cite{lieb83})
 $$
 \mathbb{V}= L^{3/2}(\mathbb{R}^{3})+L^{\infty}(\mathbb{R}^3).$$}
We observe that $\mathbb{V}_C\subset \mathbb{V}$ (see, e.g., \cite{jerison-kenig85} or Theorem XIII.57  the comment just below, page 226 of \cite{reed-simon78}), where
 \begin{eqnarray*}
\mathbb{V}_C=\left\{-\sum_{j=1}^M\frac{Z_je^2}{|x-r_j|}: Z_j\in
\mathbb{R}, r_j\in \mathbb{R}^3 (j=1,2,\cdots, M);
M=1,2,\cdots\right\}.\end{eqnarray*}
 More precisely, if $v\in \mathbb{V}_C$ and $\Psi\in \mathscr {G}_{v}$ satisfies (\ref{ground-state-weak}) and  vanishes
 on an open set of $\mathbb{R}^{3N}$, then $\Psi\equiv 0$ on $\mathbb{R}^{3N}$.

 Now we are able to prove  that the density  $\rho$ of ground state is uniquely determined by the external potential $v$, apart from an additive constant in some open set.

\begin{theorem}\label{hohenberg-kohn}
Let $\Psi_v\in \mathscr {G}_{v}$ and $\Psi_{v'}\in \mathscr
{G}_{v'}$ with $v,v'\in \mathbb{V}$. If $v$ and $v'$ differ by more than just a constant (in the sense of almost everywhere) in some open set, then
$\rho^{\Psi_v}\not=\rho^{\Psi_{v'}}$.
\end{theorem}

\begin{proof}
Let   $v$ and $v'$ differ by more than just a constant in $(a,b)^3$, where $a<b$.
If  $\rho^{\Psi_v}=\rho^{\Psi_{v'}}$, we see  from Lemma \ref{pino-etal-lemma} and the unique continuation principle  that
$$\sum^N_{i=1}(v'-v)(x_i)=E(v')-E(v)~ ~\mbox{a.e.}~ (x_1,x_2,\cdots,x_N)\in (a,b)^{3N},$$
which together with Lemma \ref{multivarible-identity} leads to
\begin{eqnarray*}\label{onevarible-identity}
(v'-v)(x)=\left(E(v')-E(v)\right)/N  ~ ~\mbox{a.e.}~ x\in (a,b)^3.
\end{eqnarray*}
This  is a contradiction to that  $v$ and $v'$ differ by more than just a constant in $(a,b)^3$. We complete the proof.
\end{proof}

%
%
%
%
%
%
%

\end{document}